\documentclass[11pt,reqno]{article}   
\usepackage{fullpage}

\usepackage[unicode=true]{hyperref}
\usepackage{amsmath}
\usepackage{cite}
\usepackage{amsfonts}
\usepackage{amssymb}
\usepackage{amsthm}
\usepackage{authblk}
\usepackage[pdftex]{color,graphicx}
\usepackage{adjustbox}
\usepackage{comment}
\usepackage{bbold}
\usepackage{tikz} 
\usetikzlibrary{decorations.pathreplacing,angles,quotes}
\usetikzlibrary{matrix,arrows,decorations.pathmorphing}
\usepackage{tikz-cd}
\usetikzlibrary{arrows}
\usetikzlibrary{decorations.markings}
\usepackage[all]{xypic}
\usepackage{bbold}
\usepackage{diagbox}
\usepackage{caption}
\usepackage{subcaption}
\usepackage{pdfpages} 
\usepackage{blkarray}
\usepackage{centernot}
\usepackage{mathtools}
\usepackage{stmaryrd}
\usepackage{soul}  
\usepackage{mypack-aziz} 

\definecolor{bluegray}{rgb}{0.4, 0.6, 0.8}
\definecolor{turquoise}{rgb}{0.2, 0.7, 0.6}

\newcommand\co[1]{{\color{black}#1}}  

\newcommand{\Facc}{\operatorname{Face}}
\newcommand{\Facet}{\operatorname{Face}}
\newcommand{\Face}{\operatorname{Face}}

\usepackage{tikz-cd}

\usepackage{todonotes}
\newcounter{commcount}
\usepackage{xcolor}

\begin{document}




\title{Extremal simplicial distributions on cycle scenarios with arbitrary outcomes}

\author[1]{Aziz Kharoof}
\author[2]{Selman Ipek}
\author[2]{Cihan Okay\footnote{Author to whom any correspondence should be addressed. Email: cihan.okay@bilkent.edu.tr}}
\affil[1]{{\small{Department of Mathematics, University of Haifa, Haifa, Israel}}}

\affil[2]{{\small{Department of Mathematics, Bilkent University, Ankara, Turkey}}}

	\maketitle


\begin{abstract}
{Cycle scenarios are a significant class of contextuality scenarios, with the Clauser-Horne-Shimony-Holt (CHSH) scenario being a notable example. While binary outcome measurements in these scenarios are well understood, the generalization to arbitrary outcomes remains less explored, except in specific cases. In this work, we employ homotopical methods in the framework of simplicial distributions to characterize all contextual vertices of the non-signaling polytope corresponding to cycle scenarios with arbitrary outcomes. Additionally, our techniques utilize the bundle perspective on contextuality and the decomposition of measurement spaces. This enables us to extend beyond scenarios formed by gluing cycle scenarios and describe contextual extremal simplicial distributions in these generalized contexts.} 
\end{abstract}

\tableofcontents  

\section{Introduction}

Simplicial distributions \cite{okay2022simplicial} {are} a \co{recent} framework {for studying contextuality} that extends the sheaf-theoretic framework of Abramsky--Brandenburger \cite{abramsky2011sheaf} and the cohomological framework of Okay et al. \cite{Coho}.
Measurements and outcomes are represented by combinatorial models of spaces known as simplicial sets, and the non-signaling conditions are encoded topologically.
The theory of simplicial distributions is an extension of the theory of non-signaling distributions, which utilizes topological ideas and tools for studying the polytopes of non-signaling distributions and the Bell polytopes describing the classical region {\cite{e25081127}}.
{In this paper, we apply homotopical techniques to fully characterize the extremal distributions on cycle scenarios with arbitrary outcomes, \co{thereby extending} the well-known binary-outcome case.} 
Our techniques generalize to more sophisticated measurement spaces obtained by gluing cycle scenarios.

A simplicial distribution consists of a family of distributions $p_\sigma$ parametrized by the simplices of the measurement space. Here $\sigma$ is an $n$-dimensional measurement (or context), and the distribution $p_\sigma$ is on the set of $n$-dimensional outcomes.
These distributions are related by the topological non-signaling conditions imposed by the simplicial structure of the measurement and outcome spaces.
In this paper, we consider a measurement {space} that is $1$-dimensional, e.g., a directed graph, whose vertices represent measurements and edges consist of pairs of measurements that can be simultaneously performed. 
A convenient way to represent {these} simplicial distribution {is} as a family of $d\times d$ matrices with entries $p_\sigma^{ab}$ for each edge $\sigma$.
The sum of the rows of this matrix gives the marginal distribution at the source vertex and the sum of the columns at the target vertex. 
A cycle scenario has the measurement space $C^{(n)}$ consisting of $n$-edges glued to form a disk's boundary. For example, $n=4$ corresponds to the famous 
Clauser--Horne--Shimony--Holt
(CHSH) scenario {\cite{chsh69}}.
Our main result is a characterization of the contextual vertices of cycle scenarios with arbitrary outcomes in $\ZZ_d=\set{0,1,\cdots,d-1}$.
A $k$-order cycle distribution, where $1\leq k\leq d$, is a particular type of distribution whose non-zero probabilities are of the form $1/k$ (Definition \ref{Def:CccyclicDis}).

\begin{thm*}
A simplicial distribution on the cycle scenario $C^{(n)}$ is a contextual vertex if and only if it is a $k$-order cycle distribution for some $k\geq 2$.
\end{thm*}

{This result provides a complete classification of extremal simplicial distributions on cycle scenarios with arbitrary outcome sets. The contextual extremal points correspond precisely to $k$-order cycle distributions with $k \geq 2$, while the remaining extremal points are deterministic and correspond to the case $k = 1$. The key novelty lies in allowing arbitrary outcome sizes: this extends the well-studied binary-outcome case and yields a full solution to the extremal problem for cycle scenarios with arbitrary outcomes.

Leveraging the explicit form of $k$-order cycle distributions, we also derive a formula for the total number of vertices of the polytope of simplicial distributions associated with the $n$-cycle scenario with $d$ outcomes (Corollary~\ref{cor:number of vertices}). For fixed $(n,d)$, the number of $k$-order cycle distributions is given by 
\[\binom{d}{k}^n (k!)^{n-1} (k-1)!.\]
Applying Stirling’s approximation to the case $k = d$ shows that the number of vertices grows super-exponentially in $d$, with leading-order behavior $\Theta((d/e)^{dn})$.}
{This rapid growth, together with the general intractability of vertex enumeration for polyhedra \cite{khachiyan2008generating}, highlights that determining all extremal points of the simplicial polytope constitutes a computationally hard enumeration problem.}


Our main method of approach follows the homotopical methods introduced in \cite{kharoof2023homotopical}.
There are two ways to extend these results:
\begin{itemize}
\item The bundle approach of \cite{barbosa2023bundle} allows us to construct scenarios where measurements can assume outcomes in different sets.  Corollary \ref{cor:kcyclicvertex} extends our theorem above to the case of arbitrary outcomes not necessarily uniform over each measurement, extending the results of {\cite{barrett2005nonlocal}}.
\item We analyze how gluing two measurement spaces affects the vertices of the associated polytope of simplicial distributions. {Theorem} \ref{thm:vertex on union} allows us to extend our results to scenarios obtained by gluing cycle scenarios.
\end{itemize}

{Non-signaling distributions, and their generalization to simplicial distributions, can be utilized for computation. A typical example is the Mermin star contextuality scenario of \cite{mermin1993hidden}, converted into a computation using the measurement-based quantum computation (MBQC) model \cite{anders2009computational}. The connection between contextuality and computational power is further analyzed in \cite{raussendorf2013contextuality}. Our techniques are applicable to Bell scenarios that can be used for such computational purposes. In particular, the $(2,3,3)$ Bell scenario---two parties, each with three trichotomic measurements---can be collapsed to a measurement space obtained by gluing two cycle scenarios. Using this method, we are able to identify a contextual vertex (Section~\ref{sec:232 Bell}). Our main classification result for cycle scenarios with arbitrary outcomes opens up new applications to MBQC where the observables have an arbitrary number of outcomes. Such extensions have been studied in the literature, for example in \cite{frembs2018contextuality}.

More recently, Bell scenarios of the form $(n,3,2)$---consisting of $n$ parties, each with three dichotomic measurements---have found applications in the classical simulation of quantum computation within the framework of quantum computation with magic states \cite{okay2024classical}.
 In this simulation model, the vertices of the polytope of simplicial distributions serve as classical keys, whose probabilistic mixtures represent the quantum state in the computation. Therefore, as a natural generalization to quantum measurements with $d$ outcomes \cite{gross2006hudson}, the enumeration of these vertices for Bell scenarios of the form $(n,d+1,d)$ seems to be crucial from the classical simulation perspective, providing a rigorous approach to identifying quantum advantage.}

The structure of the paper is as follows. In Section \ref{sec:Homotopical vertices of facets}, we introduce background material for simplicial sets and homotopical tools for analyzing {faces} of distribution polytopes.
Section \ref{sec:bundle scenarios} is about bundle scenarios and their application {to extend our analysis  to larger sets of outcomes, as demonstrated in  Proposition \ref{pro:BbundleMor}.}
In Section \ref{sec:Cycle scenario with arbitrary outcome}, we prove our main theorem highlighted above.
Gluings of cycle scenarios are analyzed in Section \ref{sec:scenarios obtained by gluing} along with interesting examples. 

\paragraph{Acknowledgments.}
This work is supported by the Air Force Office of Scientific Research (AFOSR) under
award number FA9550-21-1-0002.
The second and third authors 
acknowledge support from the Digital Horizon Europe project FoQaCiA, GA no. 101070558. {The first and third authors are also supported by AFOSR  FA9550-24-1-0257.}

\section{{Homotopical vertices of {faces}}}
\label{sec:Homotopical vertices of facets}


{We begin by recalling basic definitions from the theory of simplicial distributions {\cite{okay2022simplicial}}. Then in Section \ref{sec:The vertex support} we introduce a new notion called the vertex support based on a preorder between simplicial distributions. We use this notion to provide a characterization of vertices. In Sections \ref{sec:Null-homotopy} and \ref{sec:Vertices of facets} we recall results 
from   
\cite{kharoof2023homotopical}
on the homotopical characterization of strong contextuality}.

\subsection{Simplicial distributions}
\label{sec:Simplicial distributions}

The theory of simplicial distributions \cite{okay2022simplicial} is a combinatorial framework for describing distributions on a simplicial scenario $(X,Y)$ consisting of
\begin{itemize}
\item a measurement space $X$, and
\item an outcome space $Y$.
\end{itemize}
A space in this framework is represented by a simplicial set $X$ consisting of the data of a set of simplicies $X_0,X_1,\cdots,X_n,\cdots$ and the simplicial structure maps. 
The structure maps, given by the face maps $d_i^X:X_n\to X_{n-1}$ and degeneracy maps $s_j^X:X_n\to X_{n+1}$, encode how to glue and collapse the simplices, respectively {(see, e.g.,\cite{friedman2008elementary})}. {We usually denote the simplicial structure maps by $d_i$ and $s_j$ omitting the simplicial set from the notation.}
A simplex is called degenerate if it lies in the image of a degeneracy map. Otherwise, it is called non-degenerate. A non-degenerate simplex is called a generating simplex if it is not a face of another simplex.
A map $f:X\to Y$ between two simplicial sets is given by a collection of functions $\set{f_n:X_n\to Y_n}_{n\geq 0}$ compatible with the simplicial structure maps. We will employ the notation $f_x=f_n(x)$ for $x\in X_n$. The set of simplicial set maps will be denoted by $\catsSet(X,Y)$.

{For a set $U$, let $D(U)$ denote the set of probability distributions, i.e., functions $p:U\to \RR_{\geq 0}$ with finite support satisfying $\sum_{u\in U} p(u)=1$.}
A simplicial distribution on the scenario $(X,Y)$ is a simplicial set map
$$
p:X\to D(Y).
$$
{Here} $D(Y)$ is the simplicial set whose $n$-simplices are {given by the set $D(Y_n)$ of} distributions on $Y_n$ and the simplicial structure maps are given by marginalization along the structure maps of $Y$:
$$
{d_i^{D(Y)} = D(d_i^Y)\;\; \text{ and }\;\; s_j^{D(Y)} = D(s_j^Y).}
$$
More explicitly, a simplicial distribution consists of a family of distributions
$$
\set{p_x\in D(Y_n): x\in X_n}_{n\geq 0}
$$
compatible under the topological non-signaling conditions given by {$d_i^{D(Y)}p_x=p_{d_ix}$ and $s_j^{D(Y)}p_x =p_{s_jx}$}. We write $p_x^y$ 
for the probability {$p_x(y)$} of observing the outcome $y$ for the measurement $x$. {For {a} simplicial set map $\varphi:X \to Y$ we define the associated deterministic distribution $\delta^{\varphi}:X \to D(Y)$ {by} $\delta_Y \circ \varphi$ where $\delta_Y:Y\to D(Y)$ is the canonical map that sends a simplex to the delta distribution peaked at that simplex.} 
The set ${\catsSet(X,D(Y))}$ of simplicial distributions on $(X,Y)$ is denoted by $\sDist(X,Y)$.

Typically in applications we impose the following restrictions:
\begin{itemize}
\item $X$ is a finitely generated simplicial set, and
\item $Y$ has finitely many simplices in each dimension.
\end{itemize}
The first condition means that $X$ has finitely many generating simplices. {Under these conditions $\sDist(X,Y)$ is a polytope.}

{We are primarily interested in the following  spaces, which will serve as measurement and outcome spaces, respectively:}
\begin{itemize}
\item An {\it $n$-circle} 
is a simplicial set $C^{(n)}$ specified by 
a sequence of pairwise distinct $1$-simplices $\sigma_1,\cdots,\sigma_n$ satisfying
$
d_{0}(\sigma_i)=d_{1}(\sigma_{i+1}) 
$ for every $1 \leq i \leq n-1$ and 
$d_{0}(\sigma_n)=d_{1}(\sigma_1)$. We will   write $v_i=d_1(\sigma_i)$. 

\item For a set $U$ let $\Delta_U$ denote the simplicial set whose $n$-simplices are given by the set $U^{n+1}$ and 
	the simplicial structure maps are given by
	$$
	\begin{aligned}
		d_i(x_0,x_1,\cdots,x_n) &=( x_0,x_1,\cdots,x_{i-1},x_{i+1},\cdots,  x_n) \\
		s_j(x_0,x_1,\cdots,x_n) &=( x_0,x_1,\cdots,x_{j-1},x_j,x_{j},x_{j+1},\cdots,  x_n).
	\end{aligned}
	$$
\end{itemize}

\begin{defn}\label{def:cycle scenario}
{\rm
An \emph{$n$-cycle scenario} consists of  
\begin{itemize}
\item the measurement space given by an
{$n$-circle}
$C^{(n)}$, and
\item the outcomes space
$\Delta_{\ZZ_d}$ where $\ZZ_d=\set{0,1,\cdots,d-1}$. 
\end{itemize}
}
\end{defn}

{In Figure \ref{fig:4cycle} the $4$-circle measurement space is depicted. When the outcomes are in $\ZZ_2$ the resulting scenario is the famous CHSH scenario.}

\begin{figure}[h!] 
  \centering
  \includegraphics[width=.2\linewidth]{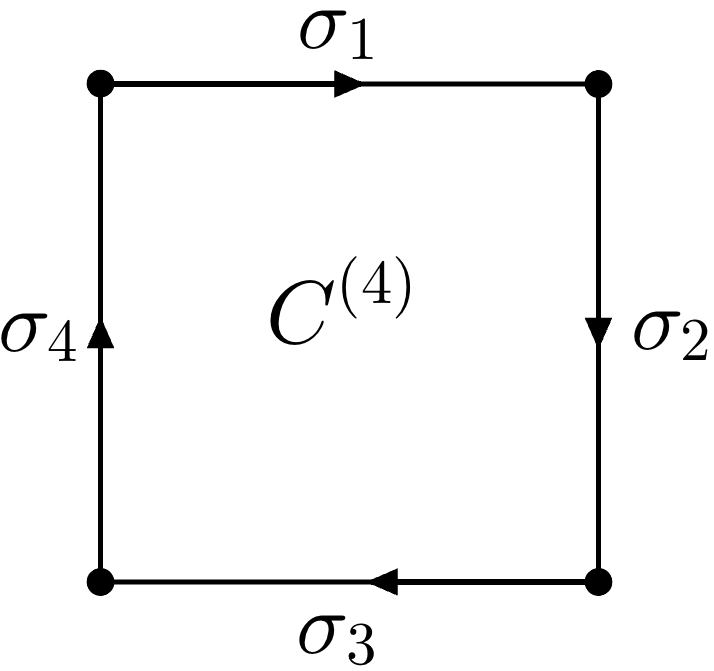}
\caption{{The $4$-circle measurement space.}}
\label{fig:4cycle}
\end{figure}

More generally, {we will consider}
$1$-dimensional 
{measurement spaces,}
i.e., basically directed graphs.
In this case the simplicial set {representing the measurement space} is specified by 
its $0$-simplices and $1$-simplices together with the face maps between them.  {We will write a distribution on an edge $\sigma$ as a {$d\times d$}-matrix $M$ where $M_{ab}=p^{ab}$. For convenience we index the rows and columns of such a matrix {by} $0,\cdots,d-1$. The marginals on the initial vertex and the terminal vertex {are} given by the sum of {the} rows and columns, respectively.}

\subsection{The vertex support}
\label{sec:The vertex support}

{Next, we introduce a preorder on the set of simplicial distributions and 
an associated
notion called the vertex support. This is used to provide a 
new
characterization of vertices.}

\begin{defn}\label{def:strong contextuality}
{\rm
The {\it support of a simplicial distribution}
	$p:X\to D(Y)$ is defined by
$$
	\supp(p) =\{ \varphi \in \St(X ,Y)  :\, p_x^{\varphi_x} \neq 0,\; \forall x\in X_n,\, n \geq 0  \}.
$$
We say $p:X\to D(Y)$ is {\it strongly contextual} if $\supp(p)$ is empty.
}
\end{defn}

Note that to have $\varphi\in \supp(p)$ it suffices that  
$p_{x}^{\varphi_x}\neq 0$ for {every} 
generating simplex $x$ of $X$.

\begin{defn}\label{def:Rrrelation}
{\rm 
Given two simplicial distributions $p,q:X \to D(Y)$ we write {$q\preceq p$} if 
$q_{x}^y\neq 0$ implies $p_x^y\neq 0$ for every $n \geq 0$, $x \in X_n$, and $y\in Y_n$.  
}
\end{defn} 
The relation $\preceq$ is a preorder on 
$\sDist(X,Y)$.
It suffices to verify the  condition for this relation only on the generating simplices of $X$. 
Recall that a simplicial distribution $p$ is called extremal if $p=\alpha q + (1-\alpha) \tilde q$ and {$0<\alpha\leq 1$} implies $p=q$. Extremal simplicial distributions are also called vertices. 
%
%
\begin{defn}\label{def:Vsuppp}
{\rm
Given a simplicial distribution $p:X \to D(Y)$. We define the vertex support $\Vsupp(p)$ to be the set of vertices 
$q\in \sDist(X,Y)$
satisfying $q \preceq p$. 
}
\end{defn}

The notion of vertex support is a natural extension of the notion of support. Note that
$\varphi \in \supp(p)$ if and only if $\delta^{\varphi} \preceq p$. Since every deterministic distribution is a vertex we have 
$\set{\delta^{\varphi}~|~\varphi \in supp(p)} \subseteq \Vsupp(p)$.
\begin{lemma}\label{lem:Extractsupp}
Given simplicial distributions $p,q:X \to D(Y)$, we have $q\preceq p$ if and only if there exists a simplicial distribution $\tilde{p}:X \to D(Y)$ and $\alpha \in (0,1]$ such that $p=\alpha q +(1-\alpha)\tilde{p}$. 
\end{lemma}
\begin{proof}
{Assume $q\preceq p$.}
We define $\alpha$ to be the minimum of the following set of real numbers
$$
\bigcup_{n \geq 0}\set{\frac{p_x^y}{q_x^y} ~|~  x \in X_n ~\text{a generator}, y\in Y_n, ~\text{and}~ q_x^y\neq 0}.
$$
Firstly, this minimum exists because the set is finite {and $\alpha > 0$ because $q\preceq p$. On the other hand, the condition that $\sum_{y\in Y_n}p_x^y=1=\sum_{y\in Y_n}q_x^y$ for every $x\in X_n$, 
implies that for every $x\in X_n$ there is $y\in Y_n$ such that $q_x^y\neq 0$ and $p_x^y \leq q_x^y$, which means that $\alpha\leq 1$}. 
If $\alpha=1$ we conclude that $p=q$,
 {in which case}
$p=1\cdot p+(1-1)p$.
For the rest we assume that $\alpha<1$. 
Then we define the simplicial distribution 
$\tilde{p}:X\to D(Y)$ by
$$
\tilde{p}_x^y=\frac{p_x^y-\alpha q_x^y}{1-\alpha}
$$
for every generator $x \in X_n$ and $y\in Y_n$. 
We have 
$$
\sum_{y\in Y_n}\tilde{p}_x^y=
\sum_{y\in Y_n}\frac{p_x^y-\alpha q_x^y}{1-\alpha}=\frac{1}{1-\alpha}(\sum_{y\in Y_n}p_x^y-\alpha \sum_{y\in Y_n}q_x^y)=
\frac{1}{1-\alpha}(1-\alpha)=1.
$$
The map $\tilde p$ is simplicial since both $p$ and $q$ are simplicial.
Finally, one can check that $p=\alpha q +(1-\alpha)\tilde{p}$.
The converse direction is clear.
\end{proof}

An immediate consequence of this result is a characterization of vertices.

\begin{cor}\label{cor:preorder vertex characterization}
A simplicial distribution $p:X\to D(Y)$ is a vertex if and only if $q\preceq p$ implies that $q=p$.
\end{cor}

\begin{pro}\label{pro:CconVertSC} 
For a simplicial distribution 
$p:X\to D(Y)$, the following properties hold:
\begin{enumerate}
\item If $p$ is a contextual vertex then $p$ is strongly contextual.
\item If $p$ is strongly contextual 
then it is contextual.
\end{enumerate}
\end{pro}
\begin{proof}
To prove (1)
suppose that $\varphi \in \supp(p)$, i.e., $\delta^{\varphi}\preceq p$. 
By Corollary \ref{cor:preorder vertex characterization} this implies that $p=\delta^\varphi$, a contradiction since $p$ is contextual.
Part 
(2) is 
proved in \cite[Pro. 2.12]{kharoof2022simplicial}.    
\end{proof}

\subsection{Null-homotopy}
\label{sec:Null-homotopy}

Simplicial sets provide a combinatorial framework for homotopy theory \cite{goerss2009simplicial}.
In this {section}, we will be restricting our attention to the notion of null-homotopy and defining it in a way best suited for our purposes. First, we introduce an important construction.

\begin{defn}(\!\!\cite{stevenson2011d})
\label{def:decalage} 
{\rm
The d\' ecalage of a simplicial set $Y$ is the simplicial set   $\Dec^0(Y)$ obtained by shifting the simplices of $Y$ down by one degree, i.e., $\Dec^0(Y)_n = Y_{n+1}$, and forgetting the first face and degeneracy maps.
}
\end{defn}  

We have 
a simplicial map $d_0: \Dec^0(Y) \to Y$
that sends a simplex to its $0$-th face. 
For $Y$ we will take the nerve space $N\ZZ_d$.
The set of $n$-simplices of the nerve space consists of $n$-tuples of elements $(a_1,a_2,\cdots,a_n)\in \ZZ_d^n$ together with the face and the degeneracy maps:
$$
\begin{aligned}
d_i(a_1,\cdots,a_n) &= \left\lbrace
\begin{array}{ll}
(a_2,\cdots,a_n) & i=0\\
(a_1,\cdots,a_i+a_{i+1},\cdots,a_n) & 0<i<n\\
(a_1,\cdots,a_{n-1}) & i=n
\end{array} 
\right.\\
s_j(a_1,\cdots,a_n) & = (a_1,\cdots,a_{j},0,a_{j+1},\cdots,a_n).
\end{aligned}
$$

\begin{defn}\label{def:Nnullhom}
{\rm
A simplicial map $\varphi:X \to N\zz_d$ is called null-homotopic if there 
exists
 a simplicial map $\psi: X \to \Dec^0
(N\zz_d)$ such that $d_0\circ \psi=\varphi$.
}
\end{defn} 

Next we provide an alternative description for the d\' ecalage of the nerve space.


\begin{lemma}\label{lem:DecDelta}
There is an isomorphism of simplicial sets
$$
\Delta_{\zz_d} \xrightarrow{\cong} \Dec^0(N\zz_d)
$$
defined in degree $m$ by sending $(a_0,a_1,\cdots,a_{m})$ to the tuple $(a_0,a_1-a_0,a_2-a_1,\cdots,a_{m}-a_{m-1})$.  
\end{lemma}
\begin{proof}
See \cite[Lemma 3.30]{kharoof2023homotopical}.
\end{proof}

Therefore the canonical map $d_0 : \Dec^0(N\zz_d) \to N\zz_d$ can be replaced by the map 
\begin{equation}\label{eq:kappa}
\kappa:\Delta_{\zz_d} \to N\zz_d
\end{equation}
defined in degree $m$ by sending $(a_0,a_1,\cdots,a_{m})$ to 
$(a_1-a_0,a_2-a_1,\cdots,a_{m}-a_{m-1})$.
\begin{prop}\label{pro:MmapNull}
The simplicial map $\varphi: C^{(n)} \to N\zz_d$ is {null-homotopic} if and only if $\sum_{i=1}^n\varphi_{\sigma_i}=0$.    
\end{prop}
\begin{proof}
See \cite[Proposition 3.4]{kharoof2023homotopical}.
\end{proof}


\subsection{{Vertices of {faces}}}
\label{sec:Vertices of facets}

Next, we define the {face} associated to a deterministic simplicial distribution and show that homotopy can detect its vertices. 
 
\begin{defn}\label{FacdefNH}
{\rm
Given $\varphi \in \catsSet(X,N\zz_d)$ we define the face $\Facc(\varphi)$ at $\varphi$ to be the preimage of 
$\delta^\varphi$
under the map 
$$
D(\kappa)_\ast:\catsSet(X,D(\Delta_{\zz_d})) \to \catsSet(X,D(N\zz_d))
$$ 
where $\kappa$ is as given in 
(\ref{eq:kappa}).
}
\end{defn}

In other words, a simplicial distribution  $p : X \to D(\Delta_{\zz_d})$ belongs to $\Facc(\varphi)$ if and only if it makes the following diagram commute
$$
\begin{tikzcd}[column sep=huge,row sep=large]
 & D(\Delta_{\zz_d}) 
\arrow[d,"D(\kappa)"]  \\
X \arrow[ru,"p"]  \arrow[r,"\delta^{\varphi}"']  & D(N\zz_d)
\end{tikzcd}
$$
{Our main result in this section 
is the homotopical detection of the vertices of {faces}.} 

\begin{prop}\label{pro:FfacetVertSC}
Given $\varphi \in \catsSet(X,N\zz_d)$, we have the following:
\begin{enumerate}
\item Every vertex in $\Facet(\varphi)$ is a vertex in 
$\sDist(X,\Delta_{\zz_d})$.
\item If $\varphi:X\to N\ZZ_d$ is not null-homotopic, then every simplicial distribution in 
$\Facc(\varphi)$ is strongly contextual.     
\end{enumerate}
\end{prop}
\begin{proof}
Part (1) follows from 
Lemma \ref{lem:DecDelta}
and
by part (1) of Proposition 2.21 in \cite{kharoof2023homotopical}.
Part (2) follows from 
Proposition 2.22 in \cite{kharoof2023homotopical}.
\end{proof}

\begin{cor}\label{cor:not null and single}
If $\varphi$ is not null-homotopic and $\Face(\varphi)=\set{p}$ then 
$p$ is a contextual vertex of $\sDist(X,\Delta_{\ZZ_d})$.
\end{cor}

\section{{Bundle} scenarios}
\label{sec:bundle scenarios}


{Bundle scenarios are first introduced in \cite{barbosa2023bundle} for developing a resource theory of contextuality in the framework of simplicial distributions. 
In this section we recall the basics of this theory and illustrate how it can be applied to study cycle scenarios with arbitrary outcomes.
}

\subsection{Simplicial distributions on bundles}

A bundle scenario consists of a simplicial set map
$$
f: E\to X
$$
where 
\begin{itemize}
\item $X$ represents the space of measurements, and
\item $E$ represents the space of events where each fiber $f^{-1}(x)$ over a simplex $x\in X_n$ represents the set of outcomes for the measurement.
\end{itemize}
This formalism is introduced in \cite{barbosa2023bundle}, where the notion is more restrictive. There, only certain kinds of simplicial set maps are termed bundle scenarios. However, the theory also works in this generality of arbitrary simplicial maps. 
With this generality, other interesting cases include
principal bundles as scenarios  
which yield
a twisted theory of simplicial distributions \cite{okay2024twisted}.
Every scenario $(X,Y)$ in the ordinary sense can be regarded as a bundle scenario {given by} the projection map ${\pr_1:}X\times Y \to {X}$. 

A simplicial distribution on the bundle scenario $f:E\to X$ is a family of distributions
$$
\set{p_x \in D(f^{-1}(x)):\, x\in X_n}_{n\geq 0}
$$
satisfying a compatibility condition induced by the simplicial structure of the measurements and the events.
More formally, a simplicial distribution $p$ on the bundle scenario $f$ is given by a commutative diagram
$$
\begin{tikzcd}[column sep=huge,row sep=large]
 & D(E) \arrow[d,"D(f)"] \\
X\arrow[ru,"p"] \arrow[r,"\delta_X"] & D(X)
\end{tikzcd}
$$
The set of simplicial distributions on $f$ is denoted by $\sDist(f)$. {For every simplicial scenario $(X,Y)$ {the set} $\sDist(X,Y)$ is in bijective correspondence with 
{$\sDist(\pr_1)$.}
}

{Here is an example of a bundle scenario which will be important for us in the next section.

\begin{example}\label{ex:BXm}
{\rm
Given a $1$-dimensional simplicial set $X$ and a 
tuple
of natural numbers $\vec{m}=(m_x)_{x\in X_0}$, where $m_x \geq 2$, 
we define the $1$-dimensional simplicial set $E(X,\vec{m})$ by setting 
\begin{itemize}
\item $E(X,\vec{m})_0=\bigsqcup_{x \in X_0} \set{x}\times\zz_{m_x}$, and
\item $E(X,\vec{m})_1=
\bigsqcup_{\sigma \in X_1} \set{\sigma}\times \zz_{m_{d_1(\sigma)}}\times \zz_{m_{d_0(\sigma)}}$ where 
\end{itemize}
$$
\begin{aligned}
d_1(\sigma,(i,j)) &=(d_1(\sigma),i)\\
 d_0(\sigma,(i,j)) &=(d_0(\sigma),j)\\
 s_0(x,i) &=(s_0(x),(i,i)).
\end{aligned} 
$$ 
The evident projection map induces a bundle scenario  
$$
f_m: E(X,\vec{m}) \to X.
$$
}
\end{example}
%
}

\subsection{{Morphisms of bundle scenarios}}

{The advantage of introducing bundle scenarios lies in the notion of morphisms between them.} 
Given bundle scenarios $f:E\to X$ and $g:E'\to {X'}$, 
a morphism $(\pi,\alpha):f\to g$ between them is given by a commutative diagram
$$
\begin{tikzcd}[column sep=huge,row sep=large]
E \arrow[d,"f"] & \pi^*(E)\arrow[l]  \arrow[r,"\alpha"] \arrow[d] & E' \arrow[d,"g"] \\
X & {X'} \arrow[l,"\pi"] \arrow[r,equal] & {X'}
\end{tikzcd} 
$$
where the left-hand {square} is a pull-back. 

Here is an example of a map between bundle scenarios {where $\pi$ is the identity map:}

\begin{example}\label{ex:Bbundlemap}
{\rm
Given injective maps $t_i:\zz_{d_1}\to \zz_{d_2}$, $1 \leq i \leq n$, we define a simplicial map $T: C^{(n)}\times \Delta_{\zz_{d_1}} \to C^{(n)}\times \Delta_{\zz_{d_2}}$ by setting
$$
T(v_i,a)=(v_i,t_i(a))
$$
where 
$a \in \zz_{d_1}$ {and $v_i$ is a vertex {of} $C^{(n)}$ (see Definition \ref{def:cycle scenario})}.
{To see that this assignment  determines the simplicial map, observe that}
$T(\sigma_i,(a,b))$ 
{is specified by}
$(\sigma_i,(t_i(a),t_{i+1}(b)))$, and for every {degenerate simplex $\sigma$ of}
$C^{(n)}$
the image of 
$(\sigma,(a_1,\cdots,a_m))$ is determined 
{using the degeneracy maps.}
%
}
\end{example}
%

\begin{defn}
{\rm
Given a morphism $(\pi,\alpha):f\to g$ between two bundle scenarios and a simplicial distribution $p$ on $f$ the push-forward distribution $q=(\pi,\alpha)_*(p)$ on $g$ is defined by
$$
q_{{x'}}^\theta = \sum_{\gamma:\alpha(\gamma,{x'})=\theta} p_{\pi({{x'}})}^\gamma
$$
where 
$\gamma\in E_n$ in the sum also satisfies that 
 $f(\gamma)=\pi({{x'}})$. 
}
\end{defn}

{Note that in the case 
{where}
$\pi$ is the identity map, we have $(\Id,\alpha)_*(p)=D(\alpha)\circ p$. }

\begin{lemma}\label{lem:CharecImi}
Consider a morphism of bundle scenarios
\begin{equation}
\begin{tikzcd}[column sep=huge,row sep=large]
E' 
\arrow[rr,hook,"i"]
\arrow[dr,"f"'] && E
\arrow[dl,"g"] \\
&  X &  
\end{tikzcd}
\end{equation}
where $i$ is injective. 
A simplicial distribution $q \in \sDist(g)$ lies in the image of $i_\ast: \sDist(f) \to \sDist(g)$ if and only if $q_x(e)=0$ for every $x \in X_n$ and $e \in E_n-\Image (i_n)$.    
\end{lemma}

\begin{proof}
Suppose we have $p \in \sDist(f)$ such that $i_\ast(p)=q$. 
Then for $x \in X_n$ and $e \in E_n-\Image (i_n)$ we have
$$
q_x(e)=(i_\ast(p))_x(e)=D(i_n)(p_x)(e)=
\sum_{e':i_n(e')=e}p_x(e').
$$
{So}
if $e \notin \Image (i_n)$ then $q_x(e)=0$.

For the converse,
suppose that $q_x(e)=0$ for every $x \in X_n$ and $e \in E_n-\Image (i_n)$. 
For 
$n \geq 0$, we define $p_n:X _n\to D(E'_n)$ to be $p_x(e')=q_x(i_n(e'))$ 
where
 $x \in X_n$ and $e' \in E'_n$. 
First we prove that $\set{p_n}_{n\geq 0}$ form a simplicial map $p :X \to D(E')$. 
Given $x \in X_n$ and $e' \in E'_{n-1}$, we have
\begin{align*} 
p_{d_k(x)}(e')=q_{d_k(x)}(i_{n-1}(e'))=D(d_k)(q_x)(i_{n-1}(e'))=
\sum_{e:\,d_k(e)=i_{n-1}(e')}q_x(e).
\end{align*}
Since
$q_x(e)=0$ 
for
$e \notin \Image (i_n)$ {and because $i_n$ is injective} we have 
$$
\begin{aligned}
\sum_{e:\,d_k(e) =i_{n-1}(e')}q_x(e)
&=\sum_{e'':\,d_k(i_n(e''))=i_{n-1}(e')}q_x(i_n(e''))\\
&=\sum_{e'':\,i_{n-1}(d_k(e''))=i_{n-1}(e')}q_x(i_n(e''))\\
& =\sum_{e'':\,d_k(e'')=e'}q_x(i_n(e'')).
\end{aligned} 
$$
On the other hand, we have
\begin{align*} 
D(d_k)(p_x)(e')=
\sum_{e'':\,d_k(e'')=e'}p_x(e'')
=\sum_{e'':\,d_k(e'')=e'}q_x(i_n(e'')).
\end{align*}
This shows that
$p_{d_k(x)}=D(d_k)(p_x)$. Analogously, one can prove that $p_{s_k(x)}=D(s_k)(p_x)$. 
Now, we prove that 
$D(i) \circ p =q$. Given $x \in X_n$ and $e \in E_n$, 
since
$i_n$ is injective we obtain 
$$
D(i_n)(p_x)(e)=\sum_{e':\,i_n(e')=e}p_x(e')=
\sum_{e':\,i_n(e')=e}q_x(i_n(e')) =\begin{cases}
    q_x(e) & \text{if} \; \; e \in \Image{(i_n)} \\
    0  & \text{else}.
\end{cases}
$$
Using this  we conclude that 
$$
D(f) \circ p {=D(g\circ i) \circ p}= D(g) \circ D(i) \circ p=D(g) \circ q=\delta_X.
$$
Therefore $p\in \sDist(f)$.
\end{proof}
%



{Our main result in this section is the preservation of some of the important features of distributions along the push-forward map.}

\begin{prop}\label{pro:BbundleMor}
Consider a morphism of bundle scenarios
\begin{equation}
\begin{tikzcd}[column sep=huge,row sep=large]
E' 
\arrow[rr,hook,"i"]
\arrow[dr,"f"'] && E
\arrow[dl,"g"] \\
&  X &  
\end{tikzcd}
\end{equation}
where $i$ is injective. 
For $p \in \sDist(f)$, we have the following:
\begin{enumerate}
\item $p$ is a deterministic distribution if and only if $i_{\ast}(p)$ is a deterministic distribution.
\item $p$ is a vertex if and only if $i_{\ast}(p)$ is a vertex.
\item $p$ is a contextual vertex if and only if $i_{\ast}(p)$ is a contextual vertex.
\end{enumerate}
\end{prop}
\begin{proof}
Part (1):
Consider
a deterministic distribution 
$\delta^{\varphi}=\delta_{E'}\circ\varphi \in \sDist(f)$, where $\varphi:X \to E'$ is a section for $f$. 
In this case $i\circ\varphi $ is a section for $g$ and 
$
i_\ast(\delta^{\varphi})=D(i)\circ \delta_{E'} \circ \varphi =
 \delta_{E}\circ i \circ \varphi= \delta^{i \circ \varphi}
$.

On the other hand, if $i_\ast(p)=\delta^{\psi}$ for some section $\psi$ for $g$, then for every $x \in X_n$ we have  
$$
1=\delta^{\psi_x}(\psi_x)=
i_\ast(p)_{x}(\psi_x)=D(i_n)(p_x)(\psi_x)=\sum_{e':\,i_n(e')=\psi_x}p_x(e').
$$
Since the map $i_n$ is injective, 
 there is a unique $e' \in E'_n$ with $i_n(e')=\psi_x$ and $p_x(e')=1$. This means that $p$ is a deterministic distribution.

Part (2): Suppose that $p \in \sDist(f)$ is a vertex. If there 
exist
 $q,s \in \sDist(g)$ and $0<\alpha <1$ such that $i_\ast(p)=\alpha q +(1-\alpha)s$, then by Lemma \ref{lem:CharecImi} we have 
$$
0=i_\ast(p)_x(e)=\alpha q_x(e) +(1-\alpha)s_x(e)
$$
for every $x \in X_n$ and $e \in E_n-\Image (i_n)$.
This implies that $q_x(e)=s_x(e)=0$. Again, by Lemma \ref{lem:CharecImi} we have $\tilde{q}$ and $\tilde{s}$ in $\sDist(f)$ such that $i_\ast(\tilde{q})=q$ and $i_\ast(\tilde{s})=s$. Therefore, using the convexity of $i_\ast$ we obtain
$$
i_\ast(p)=\alpha i_\ast(\tilde{q}) +(1-\alpha)i_\ast(\tilde{s})=i_\ast(\alpha\tilde{q} +(1-\alpha)\tilde{s}).
$$
Since $i_\ast$ is injective, we obtain that $p=\alpha\tilde{q} +(1-\alpha)\tilde{s}$. Moreover, since $p$ is a vertex we obtain that $\tilde{q}=\tilde{s}$. Therefore 
$q=i_\ast(\tilde{q})=i_\ast(\tilde{s})=s$.

Now, suppose that $i_\ast(p)$ is a vertex. By Proposition 5.3 from
\cite{barbosa2023bundle}
the map $i_\ast$ is  convex. By Proposition 5.15 \cite{kharoof2022simplicial} every vertex in $(i_\ast)^{-1}(i_{\ast}(p))$ is a vertex in $\sDist(f)$. Since $i_{\ast}$ is injective we obtain that $(i_\ast)^{-1}(i_{\ast}(p))=\set{p}$. Therefore $p$ is a vertex.

Part (3):
Follows directly from parts (1)-(2) and the fact that a vertex can either be  deterministic  or contextual. 
\end{proof}

\section{{C}ycle scenario with {arbitrary} outcome}
\label{sec:Cycle scenario with arbitrary outcome}


{In this section, we prove our main result on the classification of the contextual vertices of the polytope of simplicial distributions in cycle scenarios with arbitrary outcomes.}

\begin{defn}\label{Def:CccyclicDis}
{\rm
Let $1\leq k \leq d$. A simplicial distribution $p:C^{(n)} \to D(\Delta_{\zz_d})$ is called a \emph{$k$-order {cycle} distribution} if there exists a finite sequence 
$\vec{a}= (a_1^{(1)},\cdots,a_n^{(1)};a_1^{(2)},\cdots,a_n^{(2)};\cdots;
a_1^{(k)},\cdots,a_n^{(k)})$ of elements in $\zz_d$ such that $a_i^{(j)}\neq a_i^{(s)}$ for every $1\leq i \leq n$, $j\neq s$, and the distribution is defined by
$$
p_{\sigma_i}^{ab} = \left\lbrace
\begin{array}{ll}
\frac{1}{k} & (a,b)=(a_i^{(j)},a_{i+1}^{(j)})\; \, \text{for some $1\leq j \leq k$} \\  
0 & \text{otherwise,}
\end{array}
\right.
$$
for $1\leq i \leq n-1$, and
$$
p_{\sigma_n}^{ab} = \left\lbrace
\begin{array}{ll}
\frac{1}{k} & (a,b)=(a_n^{(j)},a_{1}^{(j+1)}) \; \, \text{for some $1\leq j \leq k-1$}\\  
\frac{1}{k} & (a,b)=(a_n^{(k)},a_{1}^{(1)})\\  
0 & \text{otherwise.}
\end{array}
\right.
$$
}
\end{defn}
Note that $1$-order cycle distributions
coincide with  deterministic distributions. 
\begin{example}{\rm
The distribution $p:C^{(2)} \to D(\Delta_{\zz_4})$   defined by 
$$
p_{\sigma_1}=\begin{pmatrix} 
	0 & \frac{1}{3} & 0 & 0 \\
 0 & 0 & 0 & 0 \\
 0 & 0 & 0 & \frac{1}{3} \\
 0 & 0 & \frac{1}{3} & 0 \\ 
	\end{pmatrix}
\;\;\text{ and } \;\;
	p_{\sigma_2}=\begin{pmatrix} 
	0 & 0 & 0 & 0\\
 0 & 0 & 0 & \frac{1}{3} \\
 0 & 0 & \frac{1}{3} & 0 \\
 \frac{1}{3} & 0 & 0 & 0 \\ 
	\end{pmatrix}$$ 
is a $3$-order cycle distribution that comes from the sequence $\vec{a}=(0,1;3,2;2,3)$.
}
\end{example}

{Our key observation is that the $k$-order cycle distributions are contextual vertices that are detected homotopically. For this we rely on the homotopical tools developed in Sections \ref{sec:Null-homotopy} and \ref{sec:Vertices of facets}.
}

\begin{prop}\label{pro:CyclicVrtex}
For $k \geq 2$, every $k$-order cycle distribution is a contextual vertex.
\end{prop}
\begin{proof}
We define a simplicial map $\varphi:C^{(n)} \to N\zz_k$ by setting
$$
\varphi_{\sigma_i}=\left\lbrace
\begin{array}{ll}
1 & i=n\\  
0 & \text{otherwise}.
\end{array}
\right.
$$
According to Proposition \ref{pro:MmapNull} the map $\varphi$ is not null-homotopic.  
Let $p$ be a distribution in $\Facet(\varphi)$, i.e., $D(\kappa)(p)=\delta^{\varphi}$. For $1\leq i\leq n-1$ we have $D(\kappa)(p)_{\sigma_i}=\delta^0$, i.e., 
$\sum_{a=0}^{k-1} p_{\sigma_i}^{aa}=1$, {which means that $p_{\sigma_i}^{ab}=0$ if $b\neq a$}. In addition,  $D(\kappa)(p)_{\sigma_n}=\delta^1$, and thus
$\sum_{a=0}^{k-1} p_{\sigma_n}^{a(a+1)}=1$, that means $p_{\sigma_n}^{ab}=0$ if $b\neq a+1$. So 
$$
p^b_{d_0(\sigma_i)}=\sum_{a=0}^{k-1}p_{\sigma_i}^{ab}= p_{\sigma_i}^{bb} \;\;\text{and}\; \;p^a_{d_1(\sigma_i)}=\sum_{b=0}^{k-1}p_{\sigma_i}^{ab}= p_{\sigma_i}^{aa},
$$
{whereas we also have}
$$
p^b_{d_0(\sigma_n)}=\sum_{a=0}^{k-1}p_{\sigma_n}^{ab}= p_{\sigma_n}^{(b-1)b} \;\;\text{and}\; \;p^a_{d_1(\sigma_n)}=\sum_{b=0}^{k-1}p_{\sigma_n}^{ab}= p_{\sigma_n}^{a(a+1)} .
$$
%
%
{Therefore}
$$
p_{\sigma_1}^{00}=p_{d_0(\sigma_1)}^0=p_{d_1(\sigma_2)}^0=p_{\sigma_2}^{00}=\cdots=p_{\sigma_{n-1}}^{00}=p_{d_0(\sigma_{n-1})}^0=p_{d_1(\sigma_n)}^0=p_{\sigma_n}^{01}
$$
and $p_{\sigma_n}^{01}=p_{d_0(\sigma_n)}^1=p_{d_1(\sigma_1)}^1=p_{\sigma_1}^{11}$. Similarly, we have
$$
p_{\sigma_1}^{11}=p_{d_0(\sigma_1)}^1=p_{d_1(\sigma_2)}^1=p_{\sigma_2}^{11}=\cdots=p_{\sigma_{n-1}}^{11}=p_{d_0(\sigma_{n-1})}^1=p_{d_1(\sigma_n)}^1=p_{\sigma_n}^{12}
$$
and $p_{\sigma_n}^{12}=p_{d_0(\sigma_n)}^2=p_{d_1(\sigma_1)}^2=p_{\sigma_1}^{22}$. At the end, we obtain that
$$
p_{\sigma_1}^{(k-1)(k-1)}=p_{d_0(\sigma_1)}^{k-1}=p_{d_1(\sigma_2)}^{k-1}=p_{\sigma_2}^{(k-1)(k-1)}=\cdots=p_{\sigma_{n-1}}^{(k-1)(k-1)}=p_{d_0(\sigma_{n-1})}^{k-1}=p_{d_1(\sigma_n)}^{k-1}=p_{\sigma_n}^{(k-1)0}
$$
and $p_{\sigma_n}^{(k-1)0}=p_{d_0(\sigma_n)}^0=p_{d_1(\sigma_1)}^0=p_{\sigma_1}^{00}$.
We conclude that $p_{\sigma_i}^{aa}=\frac{1}{k}$ for every $1 \leq i \leq n-1$ and $p_{\sigma_n}^{a(a+1)}=\frac{1}{k}$. We proved that the only distribution in $\Facet(\varphi)$ is the $k$-cycle distribution that corresponds to the following sequence
$$
(0,\cdots,0;1,\cdots,1;\cdots;k-1,\cdots,k-1).
$$
Let us denote this distribution by $q$. By
Corollary \ref{cor:not null and single}
we obtain that $q$ is a contextual vertex. 
Now, given a $k$-{order} cycle distribution  $p:C^{(n)} \to D(\Delta_{\zz_d})$ as in Definition \ref{Def:CccyclicDis}. We define $T: C^{(n)}\times \Delta_{\zz_{k}} \to C^{(n)}\times \Delta_{\zz_{d}}$ by setting
$$
T(v_i,j)=(v_i,t_i(j))
$$
where $t_i(j)=a_i^{(j)}$ (see Example \ref{ex:Bbundlemap}). The simplicial map $T$ is injective since  $a_i^{(j)}\neq a_i^{(r)}$ for every $1\leq i \leq n$, $j\neq r$. Therefore by part (3) of Proposition \ref{pro:BbundleMor} we obtain that the distribution $p=T_\ast (q)$ is a contextual vertex. 
\end{proof}
%


{
The converse of this result is also true, leading to the main result of our paper, which characterizes the vertices of cycle scenarios with arbitrary outcomes. For the converse, our main tool is the characterization of vertices using the preorder introduced in Section \ref{sec:The vertex support}.
}

\begin{thm}\label{thm:vert1}
A simplicial distribution   
$$
p:C^{(n)} \to D(\Delta_{\ZZ_d})
$$
is a contextual vertex if and only if it is a $k$-order cycle distribution for some $k \geq 2$.   
\end{thm}
\begin{proof}
Proposition \ref{pro:CyclicVrtex} shows one direction. For the other direction, 
let $p:C^{(n)} \to D(\Delta_{\zz_d})$ be a contextual vertex. 
There is $a_1,a_2 \in \zz_d$ such that $p_{\sigma_1}^{a_1a_2}\neq 0$, which means that    $p_{d_1(\sigma_2)}^{a_2}=
p_{d_0(\sigma_1)}^{a_2}\neq 0$. Therefore there exists $a_3 \in \zz_d$ such that  $p_{\sigma_2}^{a_2a_3}\neq 0$. By this process we obtain $a_1,a_2,\cdots,a_{n},a_{n+1}$ such that $p_{\sigma_i}^{a_ia_{i+1}}\neq 0$ for all $1 \leq i \leq n$.
If $a_{n+1}=a_1$ we can define the simplicial map $\varphi:C^{(n)} \to \Delta_{\zz_d}$ by {setting} $\varphi^{\sigma_i}=(a_i,a_{i+1})$ for $1\leq i\leq n$. 
Then $\varphi \in \supp(p)$, but by part (1) of Proposition \ref{pro:CconVertSC} $p$ is strongly contextual. Therefore {$a_{n+1} \neq a_1$ and we can continue the process. Then for some $k\geq 2$ we get} the following sequence
$$
a_1^{(1)},\cdots,a_n^{(1)};a_1^{(2)},\cdots,a_n^{(2)};\cdots;a_1^{(k)},\cdots,a_n^{(k)}\in \zz_d
$$
such that
\begin{enumerate}
    \item $a_i^{(j)}\neq a_i^{(s)}$ for $1\leq i \leq n$, $j\neq s$,
    \item $p_{\sigma_i}^{a_i^{(j)}a_{i+1}^{(j)}}\neq 0$ for $1\leq i \leq n-1$, $1\leq j \leq k$, 
    \item $p_{\sigma_n}^{a_n^{(j)}a_1^{(j+1)}}\neq 0$ for $1\leq j \leq k-1$, and $p_{\sigma_n}^{a_n^{(k)}a_1^{(1)}}\neq 0$.
\end{enumerate}
Let $q$ be the corresponding {$k$-order} cycle distribution for the sequence above (see Definition \ref{Def:CccyclicDis}).
 We conclude that $q \preceq p$. Then by Corollary \ref{cor:preorder vertex characterization}
we obtain that $p=q$.
\end{proof}

{
With the explicit form of the $k$-order cycle distributions we can provide a formula for the number of vertices.

\begin{cor}\label{cor:number of vertices}
The number of vertices of the polytope of simplicial distributions on the $n$-cycle scenario with $d$-outcomes is given by the formula
\[
V_{n,d} = \sum_{k=1}^d \binom{d}{k}^n (k!)^{n-1} (k-1)!. 
\]
\end{cor}
\begin{proof}
Let $V_{n,d,k}$ denote the number of $k$-order cycle distributions on the $n$-cycle scenario with $d$-outcomes. As described in Definition \ref{Def:CccyclicDis} a $k$-order distribution is specified by a $k\times n$ matrix over $\ZZ_d$:
\[
\begin{pmatrix} 
a_1^{(1)} & a_2^{(1)} & \cdots & a_n^{(1)} \\
a_1^{(2)} & a_2^{(2)} & \cdots & a_n^{(2)} \\
\vdots & \vdots & & \vdots \\
a_1^{(k)} & a_2^{(k)} & \cdots & a_n^{(k)} \\
\end{pmatrix}
\]
where in every column the entries are distinct. The number of such matrices is given by  
\[
\binom{d}{k}^n (k!)^n.
\]  
The first factor accounts for choosing $k$ distinct elements for each of the $n$ columns, and the second factor counts the number of ways to order the entries within each column.  
According to the formula for a $k$-order cycle distribution, each such matrix determines a unique distribution up to a cyclic permutation of its $k$ rows. Therefore, the number of distinct $k$-order cycle distributions is  
\[
V_{n,d,k} = \frac{\binom{d}{k}^n (k!)^n}{k} = \binom{d}{k}^n (k!)^{n-1} (k-1)!.
\]
\end{proof}
}

{Consider the} bundle scenario 
$f_m:E(C^{(n)},\vec{m}) \to C^{(n)}$ in Example \ref{ex:BXm}. 
By setting $m=\max\{m_x: \, x \in C^{(n)}_0\}$ we obtain a bundle morphism 
%
\begin{equation}
\begin{tikzcd}[column sep=huge,row sep=large]
E(C^{(n)},\vec{m}) 
\arrow[rr,hook,"i"]
\arrow[dr,"{f_m}"'] && C^{(n)} \times \Delta_{\zz_{m}}
\arrow[dl,""] \\
&  C^{(n)} &  
\end{tikzcd}
\end{equation}
where $i$ is injective.

\begin{defn}{\rm
We call a simplicial distribution $p$ on the bundle scenario $f_m$ 
a \emph{$k$-order cycle distribution} if $i_\ast(p)$ is a $k$-order cycle distribution
in the sense of Definition \ref{Def:CccyclicDis}.
}
\end{defn}

 By Theorem \ref{thm:vert1} and part $(3)$ of Proposition \ref{pro:BbundleMor} we obtain the following result.

\begin{corollary}\label{cor:kcyclicvertex}
A distribution on the bundle scenario $f_m$ 
is a contextual vertex if and only if it is a $k$-order cycle distribution for some $k \geq 2$.   
\end{corollary}

Corollary \ref{cor:kcyclicvertex}  generalizes Theorem 1 from \cite{barrett2005nonlocal}. 
{In that reference the authors describe the contextual (non-local) 
{extremal distributions (vertices)}
on {the $4$-circle.}
Our result above extends this characterization to $n$-circle spaces
with
arbitrary outcomes. 

%
%

%

\section{Scenarios obtained by gluing}
\label{sec:scenarios obtained by gluing}

In the theory of simplicial distributions an important strategy is to decompose a measurement space into smaller pieces.
In this section we {utilize this strategy to} provide a characterization which tells us when a simplicial distribution on a union of two measurement spaces is a vertex. 
Then we apply this characterization to describe vertices of scenarios obtained by gluing cycle scenarios.

{We begin by an observation on the relationship between convex decompositions and the preorder on simplicial distributions.} 
 
\begin{lemma}\label{lem:qinConv}
If $q\in \conv(\Vsupp(p))$, then $q\preceq p$.    
\end{lemma}
\begin{proof}
Assume that $q\in \conv(\Vsupp(p))$. This means that 
there exists $\alpha_1,\cdots,\alpha_n$ and $q_1,\cdots,q_n \in \Vsupp(p)$, where $\sum_{i=1}^n\alpha_i=1$, such that
$
q=\alpha_1 q_1 +\cdots+\alpha_n q_n
$. So if $q_x^{y} \neq 0$ for some $x \in X_n$ and $y\in Y_n$, then there exists $1 \leq j \leq n$ such that $(q_j)_x^y\neq 0$. This implies that $p_x^y\neq 0$.
\end{proof}

{Our decomposition result relies on this basic observation.}

\begin{pro}\label{pro:pull-back conv Vsupp}
Let $X=A\cup B$ and $p:X\to D(Y)$ be a simplicial distribution. There is a bijective correspondence between
\begin{enumerate}
\item the set of simplicial distributions $q$ on $X$ satisfying $q\preceq p$, and
\item the set of pairs $(q_A,q_B)$ of simplicial distributions  $q_A\in \conv(\Vsupp(p|_A))$ and ${q_B\in} \conv(\Vsupp(p|_B))$ satisfying $q_A|_{A\cap B}=q_B|_{A\cap B}$.
\end{enumerate}
\end{pro}
\begin{proof}
%
Let
$q: X \to D(Y)$ be  such that $q\preceq p$. 
Then $q|_{A}\preceq p|_{A}$, and because the relation $\preceq$ is transitive we have that  $\Vsupp(q|_{A}) \subseteq \Vsupp(p|_{A})$. This implies that $q|_{A} \in \conv(\Vsupp(p|_{A}))$. 
Similarly, we have that $q|_{B} \in \conv(\Vsupp(p|_{B}))$. 

Conversely, let $q_A$ and $q_B$ {be} such that 
{$q_A\in \conv(\Vsupp(p|_A))$, $q_B\in \conv(\Vsupp(p|_B))$, and} 
$q_A|_{A \cap B}=q_B|_{A \cap B}$. We can construct $q$ on $X$ by gluing $q_A$ and $q_B$. By Lemma \ref{lem:qinConv} we have {$q|_{A}\preceq p|_{A}$ and $q|_{B}\preceq p|_{B}$, which implies that} $q\preceq p$.
\end{proof}

{For us this result will be most useful in vertex detection.}

\begin{thm}\label{thm:vertex on union}
Let $X=A\cup B$.
A simplicial distribution $p:X\to D(Y)$ is a vertex if and only if $p$ is the unique {simplicial} distribution whose restrictions to $A$ and $B$ fall inside $\conv(\Vsupp(p|_A))$ and $\conv(\Vsupp(p|_B))$, respectively. 
\end{thm}
\begin{proof}
Follows immediately from Corollary \ref{cor:preorder vertex characterization}  and Proposition \ref{pro:pull-back conv Vsupp}.
\end{proof}

{Next, we apply this technique to identify the vertices of some common measurement scenarios, such as the cycle and Bell scenarios.}

\subsection{Popescu-Rohrlich boxes}
{\rm 
{We begin by applying our techniques to a well-known case:}
{T}he contextual vertices on the scenario $(C^{(4)},\Delta_{\ZZ_{2}})$ consist of the so-called Popescu-Rohrlich (PR) boxes \cite{pr94}. 
Our techniques 
 can be used to prove that PR boxes are vertices. An example of a PR box is 
{given by the}
distribution
\begin{equation}\label{eq:PR box} 
p|_{\sigma_{1}} =%
\begin{pmatrix}
0 & \frac{1}{2}\\
\frac{1}{2} & 0
\end{pmatrix}\quad\text{and}\quad%
p|_{\sigma_{i}} =%
\begin{pmatrix}
\frac{1}{2} & 0\\
0 & \frac{1}{2}
\end{pmatrix}
\end{equation}
{where}
$i=2,3,4$. 
Let us consider the decomposition $C^{(4)} = A\cup B$, where 
\begin{itemize}
\item $A$ {consists of a single edge ($1$-simplex)} $\sigma_{1}$, and 
\item $B$ {consists of} {the union} $\sigma_{2}\cup \sigma_{3} \cup \sigma_{4}$ {of three edges}{; see Figure \ref{fig:PR}.}
\end{itemize} 

\begin{figure}[h!] 
  \centering
  \includegraphics[width=.2\linewidth]{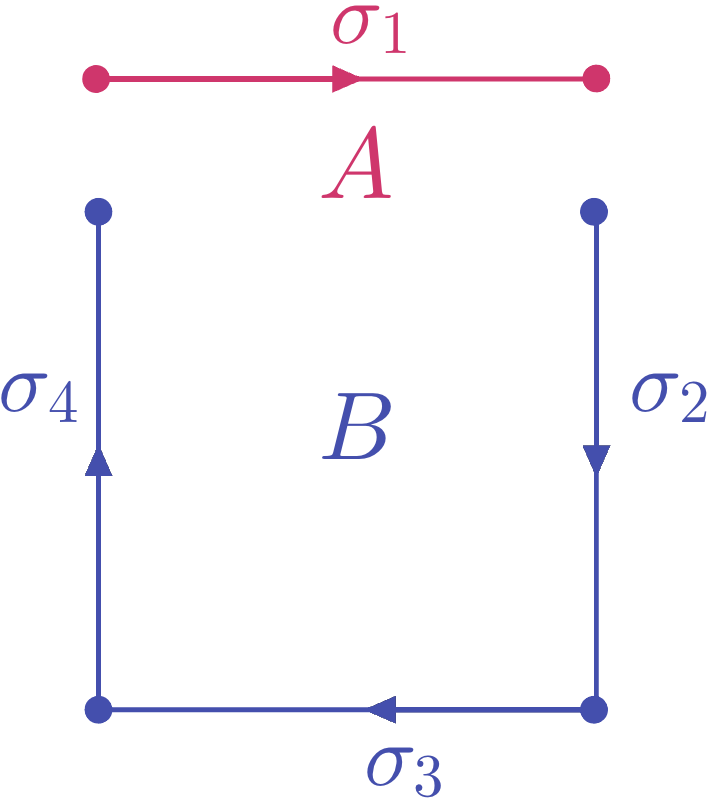}
\caption{{Decomposition of the $4$-circle.}}
\label{fig:PR}
\end{figure}

Since $A$ is a {$1$-}simplex and $B$ is a directed path obtained by gluing {along $0$-}simplices it is known (see, e.g., \cite[{Example 3.11}, Corollary 4.6]{okay2022simplicial}) that any distribution on $(A,\Delta_{\ZZ_{d}})$ and $(B,\Delta_{\ZZ_{d}})$ is noncontextual, thus the vertices are all deterministic. 
It follows that $\Vsupp(p|_{K}) {\cong} \supp(p|_{K})$ for 
$K=A,B$.
{So one can see that} 
\begin{eqnarray}
\Vsupp(p|_{A}) &=& \left \{\delta^{\varphi_{1}^{A}},\delta^{\varphi_{2}^{A}}\right \}\quad \text{where}\quad%
\delta^{\varphi_{1}^{A}} = \begin{pmatrix}0 & 1 \\ 0 & 0\end{pmatrix},~%
\delta^{\varphi_{2}^{A}} = \begin{pmatrix}0 & 0 \\ 1 & 0\end{pmatrix}
\end{eqnarray}
and
\begin{eqnarray}
\Vsupp(p|_{B}) &=& \left \{\delta^{\varphi_{1}^{B}},\delta^{\varphi_{2}^{B}}\right \}\quad \text{where}\quad%
\delta^{\varphi_{1}^{B}}|_{\sigma{i}} = \begin{pmatrix}1 & 0 \\ 0 & 0\end{pmatrix},~%
\delta^{\varphi_{2}^{B}}|_{\sigma{i}} = \begin{pmatrix}0 & 0 \\ 0 & 1\end{pmatrix},
\end{eqnarray}
{where}
$i=2,3,4$.

Let $q_{A}$ and $q_{B}$ be  distributions in $\conv\left( \Vsupp(p|_{A}) \right )$ and $\conv\left ( \Vsupp(p|_{B}) \right )$, respectively. Since there are 
two vertices in the vertex support, {these distributions} 
can be written as
\begin{eqnarray}
q_{A} = \begin{pmatrix} 0 & \alpha \\ 1-\alpha & 0 \end{pmatrix}%
\quad\quad\text{and}\quad \quad%
q_{B}|_{\sigma_{i}} = \begin{pmatrix} \beta & 0 \\ 0 & 1-\beta \end{pmatrix}%
\end{eqnarray}
for 
$i=2,3,4$ and 
$\alpha,\beta \in [0,1]$. Requiring compatibility{, i.e.,}
$$
{q_A|_{A\cap B} = q_B|_{A\cap B},}
$$
gives us that $\alpha=\beta = 1-\alpha = 1/2$. Since this yields a unique distribution {(which is the PR box in Equation (\ref{eq:PR box}))}, by {Theorem}~\ref{thm:vertex on union} we have that the PR box is a vertex of $\sDist(C^{(4)},\Delta_{\ZZ_{2}})$.

}

\subsection{A scenario with trichotomic measurements}

{\rm

We consider a simplicial distribution $p:X\to D(\Delta_{\ZZ_{3}})$ where $X$ is a $1$-dimensional space depicted in Figure \ref{fig:Example2}.
The generating simplices $\sigma_{i}$ $(i=1,\cdots,4)$ have distributions $p|_{\sigma_{i}} = Q_{i}$ where
\begin{eqnarray}
Q_{1} =%
\begin{pmatrix}%
0 & \frac{1}{3} & 0 \\
\frac{1}{3} & 0 & 0 \\
0 & 0 & \frac{1}{3}
\end{pmatrix},~~%
Q_{2} =%
\begin{pmatrix}%
\frac{1}{3} & 0 & 0 \\
0 & \frac{1}{3} & 0 \\
0 & 0 & \frac{1}{3}
\end{pmatrix},~~%
Q_{3} =%
\begin{pmatrix}%
\frac{1}{3} & 0 & 0 \\
0 & \frac{1}{3} & 0 \\
0 & 0 & \frac{1}{3}
\end{pmatrix},~~%
Q_{4} =%
\begin{pmatrix}%
0 & 0 & \frac{1}{3} \\
0 & \frac{1}{3} & 0 \\
\frac{1}{3} & 0 & 0
\end{pmatrix}.
\end{eqnarray}
We decompose our space into 
$X = A\cup B$ where 
{$A$ consists of the union  $\sigma_1 \cup \sigma_2$ and $B$ consists of 
$\sigma_3 \cup \sigma_4$.}
{Using Theorem \ref{thm:vert1} we find} the vertex supports
\begin{eqnarray}
\Vsupp(p|_{A}) = \left \{%
\begin{matrix}%
\begin{pmatrix}0 & 0 & 0\\ 0 & 0 & 0\\ 0 & 0 & 1\end{pmatrix}_{\sigma_{1}}\\
\begin{pmatrix}0 & 0 & 0\\ 0 & 0 & 0\\ 0 & 0 & 1\end{pmatrix}_{\sigma_{2}}
\end{matrix},~%
\begin{matrix}%
\begin{pmatrix}0 & \frac{1}{2} & 0\\ \frac{1}{2} & 0 & 0\\ 0 & 0 & 0\end{pmatrix}_{\sigma_{1}}\\
\begin{pmatrix}\frac{1}{2} & 0 & 0\\ 0 & \frac{1}{2} & 0\\ 0 & 0 & 0\end{pmatrix}_{\sigma_{2}}
\end{matrix}%
\right \},~
\Vsupp(p|_{B}) = \left \{%
\begin{matrix}%
\begin{pmatrix}0 & 0 & 0\\ 0 & 1 & 0\\ 0 & 0 & 0\end{pmatrix}_{\sigma_{3}}\\
\begin{pmatrix}0 & 0 & 0\\ 0 & 1 & 0\\ 0 & 0 & 0\end{pmatrix}_{\sigma_{4}}
\end{matrix},~%
\begin{matrix}%
\begin{pmatrix}\frac{1}{2} & 0 & 0\\ 0 & 0 & 0\\ 0 & 0 & \frac{1}{2}\end{pmatrix}_{\sigma_{3}}\\
\begin{pmatrix}0 & 0 & \frac{1}{2}\\ 0 & 0 & 0\\ \frac{1}{2} & 0 & 0\end{pmatrix}_{\sigma_{4}}
\end{matrix}%
\right \}.
\end{eqnarray}
{Then arbitrary} distributions $q_{A}$ and $q_{B}$ in $\conv\left (\Vsupp(p|_{A})\right )$ and $\conv\left (\Vsupp(p|_{B})\right )$, respectively, are given by
\begin{eqnarray}
q_{A} =%
\begin{matrix}%
\begin{pmatrix}0 & \frac{\alpha_{2}}{2} & 0\\ \frac{\alpha_{2}}{2} & 0 & 0\\ 0 & 0 & \alpha_{1}\end{pmatrix}_{\sigma_{1}}\\
\begin{pmatrix}\frac{\alpha_{2}}{2} & 0 & 0\\ 0 & \frac{\alpha_{2}}{2} & 0\\ 0 & 0 & \alpha_{1}\end{pmatrix}_{\sigma_{2}}
\end{matrix}%
\quad\quad\text{and}\quad\quad
q_{B} =%
\begin{matrix}%
\begin{pmatrix}\frac{\beta_{2}}{2} & 0 & 0\\ 0 & \beta_{1} & 0\\ 0 & 0 & \frac{\beta_{2}}{2}\end{pmatrix}_{\sigma_{3}}\\
\begin{pmatrix}0 & 0 & \frac{\beta_{2}}{2}\\ 0 & \beta_{1} & 0\\ \frac{\beta_{2}}{2} & 0 & 0\end{pmatrix}_{\sigma_{4}}
\end{matrix}.%
\end{eqnarray}
The distributions agree on their intersection if and only if
\begin{eqnarray}
\frac{1}{2}\alpha_{2} = \frac{1}{2}\beta_{2}\quad\quad\text{and}\quad\quad \frac{1}{2}\alpha_{2} = \beta_{1}.
\end{eqnarray}
This implies that $\beta_{1} =\beta_{2}/2 $. Using normalization $\beta_{1}+\beta_{2} = 1$ we obtain $\beta_{1} = 1/3$ and $\beta_{2} = 2/3$, which additionally yields $\alpha_{1} = 1/3$ and $\alpha_{2} = 2/3$. This is a unique solution that recovers the distribution $p$ given above. Thus $p$ is a vertex {by {Theorem}~\ref{thm:vertex on union}}.

}

\begin{figure}[h!] 
  \centering
  \includegraphics[width=.6\linewidth]{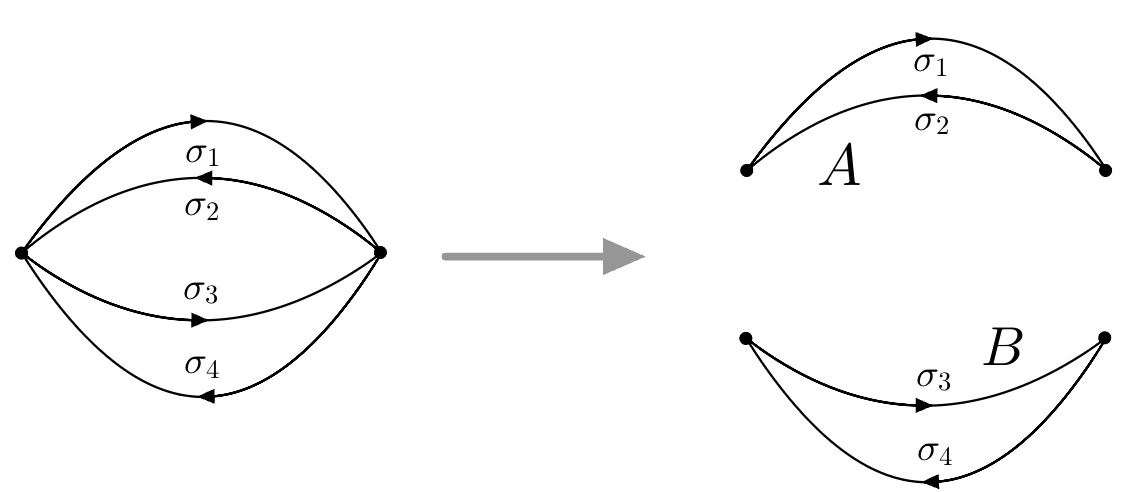}
\caption{The measurement space decomposed into two parts given by {$2$-circles.}
}
\label{fig:Example2}
\end{figure}

\subsection{$(2,3,3)$ Bell scenario}
\label{sec:232 Bell}

{\rm
{
{Next} we prove that the distribution on the $(2, 3, 3)$ Bell scenario, i.e., two parties three trichotomic measurements per party, described in Table III of \cite{jones2005interconversion} is a contextual vertex.}
Consider 
{the simplicial}
scenario $(X,\Delta_{\ZZ_{3}})$ where $X$ is a $1$-dimensional space whose underlying graph is $K_{3,3}$; {see Figure \ref{fig:Example3}.}
This corresponds to the bipartite $(2,3,3)$ Bell scenario with three measurements per party and three outcomes per measurement. There are $9$ generating simplices of $X$ that we label $\sigma_{i}$ $(i=1,\cdots,9)$. Introduce now a collection $P,Q,R,S,T$ of distributions on $(\Delta^{1},\Delta_{\ZZ_{3}})$ given by
\begin{eqnarray}
P =%
\begin{pmatrix}%
\frac{1}{4} & 0 & \frac{1}{4} \\
0 & \frac{1}{4} & 0 \\
\frac{1}{4} & 0 & 0
\end{pmatrix},\,%
Q =%
\begin{pmatrix}%
\frac{1}{4} & \frac{1}{4} & 0 \\
\frac{1}{4} & 0 & 0 \\
0 & 0 & \frac{1}{4}
\end{pmatrix},\,%
R =%
\begin{pmatrix}%
\frac{1}{4} & 0 & \frac{1}{4} \\
\frac{1}{4} & 0 & 0 \\
0 & \frac{1}{4} & 0
\end{pmatrix},\,%
S =%
\begin{pmatrix}%
0 & \frac{1}{4} & \frac{1}{4} \\
\frac{1}{4} & 0 & 0 \\
\frac{1}{4} & 0 & 0
\end{pmatrix}, \,%
T =%
\begin{pmatrix}%
\frac{1}{2} & 0 & 0 \\
0 & \frac{1}{4} & 0 \\
0 & 0 & \frac{1}{4}
\end{pmatrix}.%
\end{eqnarray}

\begin{figure}[h!] 
  \centering
  \includegraphics[width=.4\linewidth]{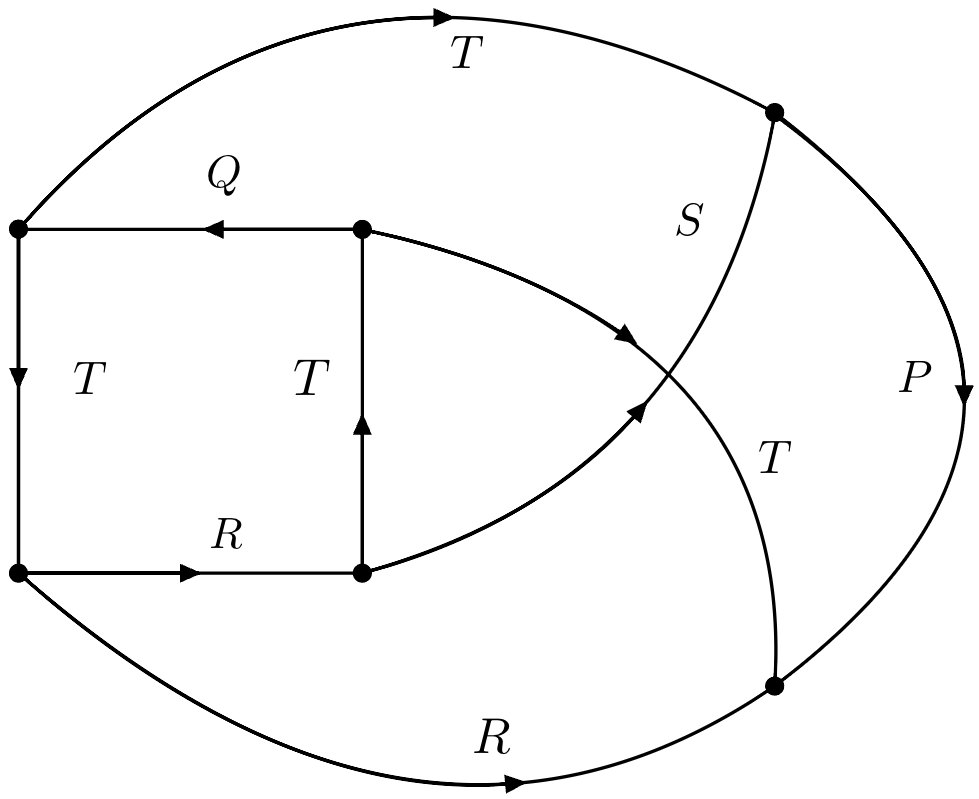}
\caption{
}
\label{fig:Example3}
\end{figure}

We consider a simplicial distribution $p:X\to\Delta_{\ZZ_{3}}$ such that $p_{\sigma_{1}} = P$, $p_{\sigma_{2}} = Q$,  $p_{\sigma_{3}} = p_{\sigma_{4}} = R$, $p_{\sigma_{5}} = S$, and $p_{\sigma_{i}} = T$ for $i=6,\cdots,9$ {(see Figure \ref{fig:Example3})}. Notice that the diagonal of $T$ sums to one, thus we can interpret $T$ as a collapsed distribution   in the sense of \cite[Section $5$]{e25081127}. 
{This means that the distribution $T$ on the edge actually comes from a point via a collapsing map. To see this let} 
$q:\Delta^{0}\to D(\Delta_{\ZZ_{3}})$ be the distribution on a point where $q = (1/2,1/4,1/4)$. {Then} we have the diagram
$$
\begin{tikzcd}[column sep=large]
    \Delta^{1} \arrow[r,"s^0"] \arrow[rr, bend left, "T"] & \Delta^{0} \arrow[r, "q"] & D(\Delta_{\ZZ_{3}})
\end{tikzcd}
$$
{where $s^0$ is the collapsing map which maps an edge to the point.}
Thus we can collapse the edges in $X$ with distributions $p_{\sigma_{i}} = T$, yielding the quotient
space $\bar X$ {depicted in Figure \ref{fig:Example3-collapse}.} 
More explicitly, we have the collapse map $\pi:X\to \bar X$ which induces a map $\pi^{\ast}:\sDist(\bar X,\Delta_{\ZZ_{3}})\to \sDist(X,\Delta_{\ZZ_{3}})$. 
{In Figure \ref{fig:Example3-collapse} we describe the distribution $\bar p \in \sDist(\bar X,\Delta_{\ZZ_{3}})$  satisfying $\pi^\ast(\bar p)=p$}. By {part 3 of} Theorem~4.4 {in} \cite{kharoof2023homotopical} {we have} that $p$ is a vertex of $\sDist(X,\Delta_{\ZZ_{3}})$ if and only if $\bar p$ is a vertex of $\sDist(\bar X,\Delta_{\ZZ_{3}})$. Moreover, it is clear that $\bar p$ is a vertex if and only if $q:Z\to \Delta_{\ZZ_{3}}$ is a vertex, where $Z$ {and the distribution $q$ {are} given {as} {in} Figure \ref{fig:Example3-collapse-b}.} 

\begin{figure}[h!]
\centering
\begin{subfigure}{.33\textwidth}
  \centering
  \includegraphics[width=.6\linewidth]{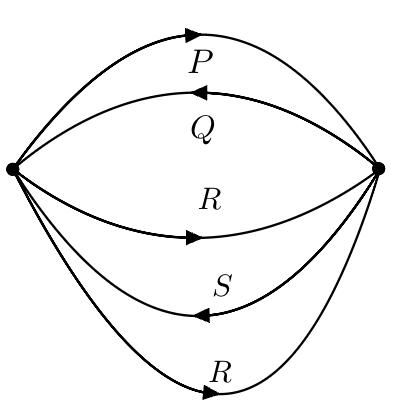}
  \caption{}
  \label{fig:Example3-collapse}
\end{subfigure}%
\begin{subfigure}{.33\textwidth}
  \centering
  \includegraphics[width=.6\linewidth]{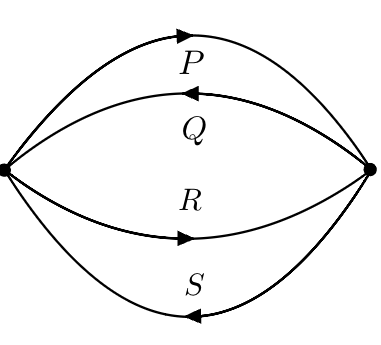}
  \caption{}
  \label{fig:Example3-collapse-b}
\end{subfigure}
\begin{subfigure}{.33\textwidth}
  \centering
  \includegraphics[width=.6\linewidth]{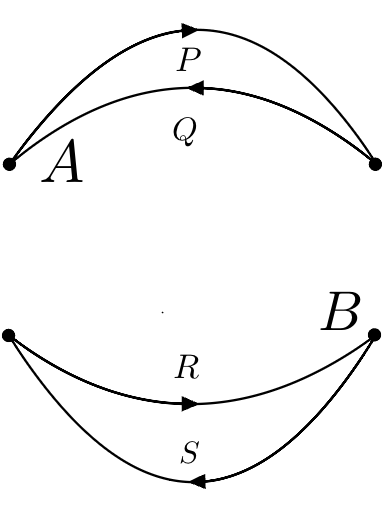}
  \caption{}
  \label{fig:Example3-decomp}
\end{subfigure}
\caption{{(a) Space obtained from $K_{3,3}$ by collapsing the edges with distribution $T$ to a point. (b) {One of the two copies of the edge with distribution $R$ is removed.} 
(c) {{Measurement} space decomposed into two } {parts}.} 
}
\label{fig:mermin-scenario-and-os}
\end{figure}

We are now in a position to apply {Theorem}~\ref{thm:vertex on union}. Consider the decomposition $Z = A \cup B$, where each piece is a $2$-circle $C^{(2)}$ {as in Figure \ref{fig:Example3-decomp}.}
The corresponding vertex supports are given by
\begin{eqnarray}
\Vsupp(q|_{A}) &=& \left \{%
\begin{matrix}%
\begin{pmatrix}1 & 0 & 0\\ 0 & 0 & 0\\ 0 & 0 & 0\end{pmatrix}_{\sigma_{1}}\\
\begin{pmatrix}1 & 0 & 0\\ 0 & 0 & 0\\ 0 & 0 & 0\end{pmatrix}_{\sigma_{2}}
\end{matrix},~%
\begin{matrix}%
\begin{pmatrix}\frac{1}{2} & 0 & 0\\ 0 & \frac{1}{2} & 0\\ 0 & 0 & 0\end{pmatrix}_{\sigma_{1}}\\
\begin{pmatrix}0 & \frac{1}{2} & 0\\ \frac{1}{2} & 0 & 0\\ 0 & 0 & 0\end{pmatrix}_{\sigma_{2}}
\end{matrix},~%
\begin{matrix}%
\begin{pmatrix}0 & 0 & \frac{1}{3}\\ 0 & \frac{1}{3} & 0\\ \frac{1}{3} & 0 & 0\end{pmatrix}_{\sigma_{1}}\\
\begin{pmatrix}0 & \frac{1}{3} & 0\\ \frac{1}{3} & 0 & 0\\ 0 & 0 & \frac{1}{3}\end{pmatrix}_{\sigma_{2}}
\end{matrix},~%
\begin{matrix}%
\begin{pmatrix}0 & 0 & \frac{1}{2}\\ 0 & 0 & 0\\ \frac{1}{2} & 0 & 0\end{pmatrix}_{\sigma_{1}}\\
\begin{pmatrix}\frac{1}{2} & 0 & 0\\ 0 & 0 & 0\\ 0 & 0 & \frac{1}{2}\end{pmatrix}_{\sigma_{2}}
\end{matrix}
\right \}\\
~\notag\\
\Vsupp(q|_{B}) &=& \left \{%
\begin{matrix}%
\begin{pmatrix}0 & 0 & 0\\ 1 & 0 & 0\\ 0 & 0 & 0\end{pmatrix}_{\sigma_{4}}\\
\begin{pmatrix}0 & 1 & 0\\ 0 & 0 & 0\\ 0 & 0 & 0\end{pmatrix}_{\sigma_{5}}
\end{matrix}~,%
\begin{matrix}%
\begin{pmatrix}0 & 0 & 1\\ 0 & 0 & 0\\ 0 & 0 & 0\end{pmatrix}_{\sigma_{4}}\\
\begin{pmatrix}0 & 0 & 0\\ 0 & 0 & 0\\ 1 & 0 & 0\end{pmatrix}_{\sigma_{5}}
\end{matrix}~,%
\begin{matrix}%
\begin{pmatrix}\frac{1}{2} & 0 & 0\\ 0 & 0 & 0\\ 0 & \frac{1}{2} & 0\end{pmatrix}_{\sigma_{4}}\\
\begin{pmatrix}0 & 0 & \frac{1}{2}\\ \frac{1}{2} & 0 & 0\\ 0 & 0 & 0\end{pmatrix}_{\sigma_{5}}
\end{matrix}~
\right \}.
\end{eqnarray}
General distributions $\rho_{A}$ and $\rho_{B}$ on the faces defined by $\Vsupp(q|_{A})$ and $\Vsupp(q|_{B})$, respectively, are given by
\begin{eqnarray}
\rho_{A} =%
\begin{matrix}%
\begin{pmatrix}\alpha_{1}+\frac{\alpha_{2}}{2} & 0 & \frac{\alpha_{3}}{3}+\frac{\alpha_{4}}{2}\\
0 & \frac{\alpha_{2}}{2}+\frac{\alpha_{3}}{3} & 0\\
\frac{\alpha_{3}}{3}+\frac{\alpha_{4}}{2} & 0 & 0%
\end{pmatrix}_{\sigma_{1}}\\
\begin{pmatrix}%
\alpha_{1}+\frac{\alpha_{4}}{2} & \frac{\alpha_{2}}{2}+\frac{\alpha_{3}}{3} & 0\\
\frac{\alpha_{2}}{2}+\frac{\alpha_{3}}{3} & 0 & 0\\
0 & 0 & \frac{\alpha_{3}}{3}+\frac{\alpha_{4}}{2}%
\end{pmatrix}_{\sigma_{2}}
\end{matrix}%
\quad\text{and}\quad
\rho_{B}=%
\begin{matrix}%
\begin{pmatrix}%
\frac{\beta_{3}}{2} & 0 & \beta_{2}\\
\beta_{1} & 0 & 0\\
0 & \frac{\beta_{3}}{2} & 0%
\end{pmatrix}_{\sigma_{4}}\\
\begin{pmatrix}%
0 & \beta_{1} & \frac{\beta_{3}}{2}\\
\frac{\beta_{3}}{2} & 0 & 0\\
\beta_{2} & 0 & 0%
\end{pmatrix}_{\sigma_{5}}
\end{matrix},%
\end{eqnarray}
where $\alpha_{i},\beta_{j}\in [0,1]$ and $\sum_{i=1}^{4}\alpha_{i} = \sum_{j=1}^{3}\beta_{j} = 1$. The two distributions agree on their intersection if and only if the following system of equations holds
\begin{eqnarray}\label{eq:systemEq}
\beta_{1} = \frac{1}{2}\alpha_{2}+\frac{1}{3}\alpha_{3},\quad\quad \beta_{3} = \alpha_{2}+\frac{2}{3}\alpha_{3},\quad\quad\beta_{3} = \frac{2}{3}\alpha_{3}+\alpha_{4},\quad\quad \beta_{2} = \frac{1}{3}\alpha_{3}+\frac{1}{2}\alpha_{4}.
\end{eqnarray}
Setting the middle two equations equal implies that $\alpha_{2} = \alpha_{4}$ that in turn implies that $\beta_{1} = \beta_{2} = \beta_{3}/2$, from which normalization gives {$\beta_{1} =\beta_{2}= 1/4$ and $\beta_{3}= 1/2$}. 
{We substitute {these} into Equation (\ref{eq:systemEq}) to {obtain} that 
$$
\frac{\alpha_{2}}{2}+\frac{\alpha_{3}}{3}=\frac{\alpha_{3}}{3}+\frac{\alpha_{4}}{2}=\frac{\alpha_{3}}{3}+\frac{\alpha_{4}}{2}=1/4 .
$$
%
Together with the normalization we obtain that $\rho_{A}|_{\sigma_1}=P$ and $\rho_{A}|_{\sigma_2}=Q$.}
The uniqueness of the solution implies that $q$ is a vertex by {Theorem}~\ref{thm:vertex on union}, which then implies that $p$ is a vertex.
}

\section{Conclusion}

{Our main
contribution is to characterize all contextual extremal simplicial distributions of cycle scenarios with arbitrary
outcomes by solving} the vertex enumeration problem for the polytope of simplicial distributions on cycle scenarios with arbitrary outcomes. 
The simplicial distribution formulation aligns with the usual non-signaling formulation. 
{Our result extends the known classification for binary outcomes to the case of arbitrary outcomes.}
Our approach uniquely benefits from novel homotopical methods, extending to scenarios formed by gluing cycle scenarios together. This gluing process is a crucial aspect of the theory of simplicial distributions when studying contextual distributions.

Decomposing a measurement space into smaller components is a notable feature of our work that warrants further investigation. As illustrated in Section \ref{sec:scenarios obtained by gluing}, if the vertex enumeration problem (see, e.g., \cite{avis1991pivoting}) is solved for measurement spaces $A$ and $B$ in a given outcome space, this solution can be leveraged to detect vertices in the composite space $X=A\cup B$. Consequently, our results suggest a novel algorithm for enumerating the vertices of polytopes derived from simplicial distributions. Systematizing this procedure to identify additional classes of contextual vertices in other scenarios is a direction for future research. 

{Moreover, understanding contextual extremal distributions has applications to quantum advantage in the scheme of measurement-based quantum computation \cite{anders2009computational,raussendorf2013contextuality, hoban2011generalized}. Our classification can be utilized to prove separation results in computational complexity, e.g., generalizing and extending results based on elevating a parity computer to a classically universal one. {In addition, recent work has demonstrated that the vertices of polytopes of simplicial distributions arising from Bell scenarios play a central role in the classical simulation of quantum computation within the magic-state model \cite{okay2024classical}.} Simplicial distributions provide the essential tools to investigate this kind of problem by providing otherwise impossible classification results, as we demonstrate in this paper. Our work opens the way for new results and brings a novel topological approach.}

\bibliography{bib.bib}

\begin{thebibliography}{10}

\bibitem{okay2022simplicial}
C.~Okay, A.~Kharoof, and S.~Ipek, ``Simplicial quantum contextuality,'' {\em
  Quantum}, vol.~7, 2023.
\newblock doi:
  \href{https://doi.org/10.22331/q-2023-05-22-1009}{10.22331/q-2023-05-22-1009}.

\bibitem{abramsky2011sheaf}
S.~Abramsky and A.~Brandenburger, ``The sheaf-theoretic structure of
  non-locality and contextuality,'' {\em New Journal of Physics}, vol.~13,
  no.~11, p.~113036, 2011.
\newblock doi:
  \href{https://doi.org/10.1088/1367-2630/13/11/113036}{10.1088/1367-2630/13/11/113036}.

\bibitem{Coho}
C.~Okay, S.~Roberts, S.~D. Bartlett, and R.~Raussendorf, ``Topological proofs
  of contextuality in quantum mechanics,'' {\em Quantum Information \&
  Computation}, vol.~17, no.~13-14, pp.~1135--1166, 2017.
\newblock doi:
  \href{https://doi.org/10.26421/QIC17.13-14-5}{10.26421/QIC17.13-14-5}. arXiv:
  \href{https://arxiv.org/abs/1701.01888}{1701.01888}.

\bibitem{e25081127}
A.~Kharoof, S.~Ipek, and C.~Okay, ``Topological methods for studying
  contextuality: N-cycle scenarios and beyond,'' {\em Entropy}, vol.~25, no.~8,
  2023.
\newblock doi: \href{https://doi.org/10.3390/e25081127}{10.3390/e25081127}.

\bibitem{chsh69}
J.~F. Clauser, M.~A. Horne, A.~Shimony, and R.~A. Holt, ``Proposed experiment
  to test local hidden-variable theories,'' {\em Physical review letters},
  vol.~23, no.~15, p.~880, 1969.
\newblock doi:
  \href{https://link.aps.org/doi/10.1103/PhysRevLett.23.880}{10.1103/PhysRevLett.23.880}.

\bibitem{khachiyan2008generating}
L.~Khachiyan, E.~Boros, K.~Borys, V.~Gurvich, and K.~Elbassioni, ``Generating
  all vertices of a polyhedron is hard,'' in {\em Twentieth Anniversary Volume:
  Discrete \& Computational Geometry}, pp.~1--17, Springer, 2008.

\bibitem{kharoof2023homotopical}
A.~Kharoof and C.~Okay, ``Homotopical characterization of strongly contextual
  simplicial distributions on cone spaces,'' {\em arXiv preprint
  arXiv:2311.14111}, 2023.

\bibitem{barbosa2023bundle}
R.~S. Barbosa, A.~Kharoof, and C.~Okay, ``A bundle perspective on
  contextuality: Empirical models and simplicial distributions on bundle
  scenarios,'' {\em arXiv preprint arXiv:2308.06336}, 2023.

\bibitem{barrett2005nonlocal}
J.~Barrett, N.~Linden, S.~Massar, S.~Pironio, S.~Popescu, and D.~Roberts,
  ``Nonlocal correlations as an information-theoretic resource,'' {\em Physical
  Review A}, vol.~71, p.~022101, Feb 2005.
\newblock doi:
  \href{https://link.aps.org/doi/10.1103/PhysRevA.71.022101}{10.1103/PhysRevA.71.022101}.
  arXiv: \href{https://arxiv.org/abs/quant-ph/0404097}{quant-ph/0404097}.

\bibitem{mermin1993hidden}
N.~D. Mermin, ``Hidden variables and the two theorems of {J}ohn {B}ell,'' {\em
  Reviews of Modern Physics}, vol.~65, no.~3, p.~803, 1993.
\newblock doi:
  \href{https://link.aps.org/doi/10.1103/RevModPhys.65.803}{10.1103/RevModPhys.65.803}.

\bibitem{anders2009computational}
J.~Anders and D.~E. Browne, ``Computational power of correlations,'' {\em
  Physical Review Letters}, vol.~102, no.~5, p.~050502, 2009.

\bibitem{raussendorf2013contextuality}
R.~Raussendorf, ``Contextuality in measurement-based quantum computation,''
  {\em Physical Review A—Atomic, Molecular, and Optical Physics}, vol.~88,
  no.~2, p.~022322, 2013.

\bibitem{frembs2018contextuality}
M.~Frembs, S.~Roberts, and S.~D. Bartlett, ``Contextuality as a resource for
  measurement-based quantum computation beyond qubits,'' {\em New Journal of
  Physics}, vol.~20, no.~10, p.~103011, 2018.

\bibitem{okay2024classical}
C.~Okay, A.~T. Yucel, and S.~Ipek, ``Classical simulation of universal
  measurement-based quantum computation using multipartite bell scenarios,''
  {\em arXiv preprint arXiv:2410.23734}, 2024.

\bibitem{gross2006hudson}
D.~Gross, ``Hudson’s theorem for finite-dimensional quantum systems,'' {\em
  Journal of mathematical physics}, vol.~47, no.~12, 2006.

\bibitem{friedman2008elementary}
G.~Friedman, ``An elementary illustrated introduction to simplicial sets,''
  {\em arXiv preprint arXiv:0809.4221}, 2008.
\newblock doi:
  \href{https://doi.org/10.48550/arXiv.0809.4221}{10.48550/arXiv.0809.4221}.

\bibitem{kharoof2022simplicial}
A.~Kharoof and C.~Okay, ``Simplicial distributions, convex categories and
  contextuality,'' {\em arXiv preprint arXiv:2211.00571}, 2022.

\bibitem{goerss2009simplicial}
P.~G. Goerss and J.~F. Jardine, {\em Simplicial homotopy theory}.
\newblock Springer Science \& Business Media, 2009.

\bibitem{stevenson2011d}
D.~Stevenson, ``D\'{e}calage and {K}an's simplicial loop group functor,'' {\em
  Theory Appl. Categ.}, vol.~26, pp.~No. 28, 768--787, 2012.

\bibitem{okay2024twisted}
C.~Okay and W.~H. Stern, ``Twisted simplicial distributions,'' {\em arXiv
  preprint arXiv:2403.19808}, 2024.

\bibitem{pr94}
S.~Popescu and D.~Rohrlich, ``Quantum nonlocality as an axiom,'' {\em
  Foundations of Physics}, vol.~24, no.~3, pp.~379--385, 1994.
\newblock doi: \href{https://doi.org/10.1007/BF02058098}{10.1007/BF02058098}.

\bibitem{jones2005interconversion}
N.~S. Jones and L.~Masanes, ``Interconversion of nonlocal correlations,'' {\em
  Physical Review A}, vol.~72, no.~5, p.~052312, 2005.

\bibitem{avis1991pivoting}
D.~Avis and K.~Fukuda, ``A pivoting algorithm for convex hulls and vertex
  enumeration of arrangements and polyhedra,'' in {\em Proceedings of the
  seventh annual symposium on Computational geometry}, pp.~98--104, 1991.

\bibitem{hoban2011generalized}
M.~J. Hoban, J.~J. Wallman, and D.~E. Browne, ``Generalized bell-inequality
  experiments and computation,'' {\em Physical Review A—Atomic, Molecular,
  and Optical Physics}, vol.~84, no.~6, p.~062107, 2011.

\end{thebibliography}
\bibliographystyle{ieeetr}

\end{document}